\documentclass[twoside,11pt]{article}

\usepackage{jmlr2e}

\providecommand{\qedsymbol}{$\square$}
\newcommand{\mathqed}{\quad\hbox{\qedsymbol}}
\DeclareRobustCommand{\qedhere}{%
  \ifmmode \mathqed
  \else
    \leavevmode\unskip\penalty9999 \hbox{}\nobreak\hfill
    \quad\hbox{\qedsymbol}%
  \fi
}

\usepackage{amsfonts}       
\usepackage{graphicx}
\usepackage{amsmath}
\usepackage{amssymb}
\usepackage{bbm}
\usepackage{mathptmx}
\usepackage{mathrsfs}
\usepackage{enumitem}

 \usepackage{color}
\usepackage{natbib}

\newcommand{\QL}{\mathsf{QL}}
\newcommand{\Proj}{\mathsf{Proj}}

\newcommand{\D}{\mathbb{D}^{m}} 
\newcommand{\PO}{\mathbb{P}^{m}} 
\newcommand{\Dn}{\mathbb{D}^{n}} 
\newcommand{\DP}{\mathbb{D}^{m-1}} 
\newcommand{\CH}{\mathbb{C}_h^{n\times n}} 
\newcommand{\CL}{\mathbb{C}_h^{mn\times mn}} 

\newcommand{\CP}{\mathbb{C}_h^{m\times m}} 

\newcommand{\diag}{\text{Diag}}

\newcommand{\marg}{{\mathsf{marg}}}

\newcommand{\pspace}{\varOmega}

\newcommand{\reals}{\mathbb{R}}

\newcommand{\gambles}{\mathcal{H}}

\newcommand{\domain}{\mathcal{K}}

\newcommand{\rdesirs}{\mathcal{K}}

\newcommand{\cone}{\mathcal{C}}

\renewcommand{\epsilon}{\varepsilon}
\newcommand{\eps}{\epsilon}

\renewcommand{\iff}{\Leftrightarrow}

\ShortHeadings{Quantum rational preferences and desirability}{Benavoli, Facchini, Zaffalon}
\firstpageno{1}

\begin{document}

%

\title{Quantum rational preferences and desirability}

\author{\name Alessio Benavoli\thanks{The authors are listed in alphabetic order.}  \email alessio@idsia.ch \\
       \name Alessandro Facchini \email alessandro.facchini@idsia.ch \\
       \name Marco Zaffalon \email zaffalon@idsia.ch\\
       \addr Istituto Dalle Molle di Studi sull'Intelligenza Artificiale (IDSIA). \\ 
       Galleria 2, Via Cantonale 2c  \\
6928 Manno, Switzerland}


\editor{}

\maketitle

\begin{abstract}
We develop a theory of quantum rational decision making in the tradition of Anscombe and Aumann's axiomatisation of preferences on horse lotteries. It is essentially the Bayesian decision theory generalised to the space of Hermitian matrices. Among other things, this leads us to give a representation theorem showing that quantum complete rational preferences are obtained by means of expected utility considerations. 
\end{abstract}


\section{Introduction}

The aim of this paper is simple. We have recently shown in \citet{benavoli2016d} that Quantum Mechanics (QM) coincides with the Bayesian theory once this is formulated in the space of Hermitian matrices (so as to make possible to gamble on quantum experiments). Such an identification makes a number of things, at least in principle, straightforward: one of these is the extension of QM to make it handle non-linear utility. We do so by adapting the traditional axiomatisation of rational preferences by \cite{anscombe1963} to the quantum setting.

After axiomatising quantum rational preferences in this way, we move on to give a representation theorem that shows that quantum rational preferences can be obtained by means of expected utility considerations. Our route to this results is based on the approach devised by \cite{zaffalon2015desirability} in the classical case: we show that the axiomatisation of quantum rational preferences is equivalent to the quantum theory of coherent desirable gambles---the very same theory at the heart of \citeauthor{benavoli2016d}'s~\citeyearpar{benavoli2016d} formulation of QM---yet formulated so as to consider prizes other then events. Intuitively, this allows us to formally bring quantum rational preferences back to plain QM through an enlargement of the space of possibilities. We eventually show how this leads to quantum probabilities and utilities after enforcing axioms for state independence and completeness.

%
%
Before briefly reviewing related work and making some final comments, we illustrate one elegant consequence of the correspondence between preferences and desirability: namely, how to derive a coherent rule for updating preferences determining how should a subject rationally change her preferences in the prospect of obtaining new information in the form of an event.

Since the wording ``quantum'' is used nowadays with a number of acceptations in the literature, we would like to make precise what our framework actually addresses: that is, nothing else but gambling on quantum mechanics experiments; we are not, at this stage, endorsing any other interpretation of the quantum decision theory developed here. Moreover, we would like to remark that our framework is an actual generalisation of classical decision-theoretic approaches in the tradition of~\cite{anscombe1963}: in fact, we can recover them by simply focusing on the subset of Hermitian matrices made by diagonal real-valued matrices; that is, by focusing on classical experiments.

%
%
%
%
%
%
%
%



\section{Rational gambling on quantum experiments}\label{sec:gambling}
We start by defining a gambling system about the results of a quantum experiment. 
To this end, we consider two subjects: the bookmaker and the gambler (Alice).  The bookmaker prepares the quantum system in some quantum state. Alice
has her personal knowledge (beliefs) about the experiment---possibly no knowledge at all.

\begin{enumerate}
\item   {
The bookmaker 
 announces that he will measure the quantum system along  its $n$ orthogonal directions and so the outcome of the measurement is an element of $\Omega=\{\omega_1,\dots,\omega_n\}$,  with $\omega_i$ denoting the elementary event ``detection along $i$''. 
Mathematically,  it means that the quantum system is measured along its eigenvectors,\footnote{We mean the eigenvectors of the density matrix of the quantum 
system.} i.e., the projectors\footnote{A projector $\Pi$ is a set of $n$ positive semi-definite matrices in $\CH$ s.t.\   $\Pi_i\Pi_k=0$, $(\Pi_i)^2=\Pi_i=(\Pi_i)^\dagger$,  $\sum_{i=1}^n \Pi_i=I$.} $\Pi^*=\{\Pi^*_{1},\dots,\Pi^*_{n}\}$
and $\omega_i$ is the event ``indicated'' by the $i$-th projector. The bookmaker is fair, meaning that he will correctly perform the experiment and report
 the actual results to Alice.} 
  \item {Before the experiment, Alice declares the set of gambles she is willing to accept.  Mathematically, a gamble $G$ on this experiment 
is a Hermitian matrix in $\CH$, the space of all Hermitian  $n \times n$ matrices being denoted by $\CH$.  We will denote the set of gambles Alice is willing to accept by $\domain \subseteq \CH$.}
\item By accepting  a gamble $G$, Alice commits herself to receive  $\gamma_{i}\in \reals$ utiles  if the outcome of the experiment eventually happens to be 
$\omega_i$. The value $\gamma_{i}$ is defined from $G$ and $\Pi^{*}$ through $ \Pi^{*}_{i}G\Pi^{*}_{i}=\gamma_{i}\Pi^{*}_{i} \text{ for } i=1,\dots,n$. 
It is a real number since $G$ is Hermitian.
\end{enumerate}
Denote by $\gambles^+=\{G\in\CH:G\gneq0\}$ the
subset of all positive semi-definite and non-zero (PSDNZ) matrices  in $\CH$: we call them the set of \emph{positive gambles}.
The set of negative gambles is similarly given by $\gambles^-=\{G\in\CH:G\lneq0\}$.  Alice examines the gambles in $\CH$ and comes up with the subset $\domain$ of the gambles that she finds desirable. Alice's rationality is characterised as follows:
\begin{enumerate}[noitemsep,nolistsep]
 \item Any gamble $G \in \CH$ such that $G \gneq0$ must be desirable for Alice, given that it may increase Alice's 
utiles without ever decreasing them
 (\textbf{accepting partial gain}). This means that $ \gambles^+ \subseteq \domain$.
\item Any gamble $G \in \CH$ such that $G \lneq0$ must not be desirable for Alice, given that it may only decrease 
Alice's utiles without ever increasing them  (\textbf{avoiding partial loss}). This means that $ \gambles^- \cap 
\domain=\emptyset$.
\item If Alice finds $G$ desirable, that is
$G \in \domain$, then also $\nu G$ must be desirable for her for any $0<\nu \in \reals$ (\textbf{positive homogeneity}).
\item If Alice finds $G_1$ and $G_2$ desirable, that is
$G_1,G_2 \in \domain$, then she must also accept $G_1+G_2$, i.e., $G_1+G_2 \in \domain$ (\textbf{additivity}). 
\end{enumerate}
To understand these rationality criteria, originally presented in \citet[Sec. III]{benavoli2016d}, we must remember that mathematically the payoff for any gamble $G$
is computed as $\Pi_i^{*} G \Pi_i^{*}$ if the outcome of the experiment is the event indicated by $\Pi_i^{*}$.
Then the first two rationality criteria above hold no matter the experiment $\Pi^{*}$ that 
is eventually performed. 
In fact,  from the properties of PSDNZ matrices, if 
 $G \gneq0$ then  $\Pi_i^{*} G \Pi_i^{*}=\gamma_{i} \Pi_i^{*}$ with $\gamma^{*}_{i}\geq0$ for any $i$ and 
$\gamma_{j}>0$ for some $j$. Therefore, by accepting $G \gneq0$, Alice can only increase her utiles.
Symmetrically,  if $G \lneq0$ then $\Pi_i^{*} G \Pi_i^{*} = \gamma_{i} 
\Pi_i^{*}$ with $\gamma_{i}\leq 0$ for any $i$. 
Alice must then avoid the gambles $G \lneq0$ because they can only decrease her utiles.
This justifies  the first two rationality criteria. 
 For the last two, observe that 
 $$
 \Pi_i^{*} (G_1+G_2) \Pi_i^{*}=\Pi_i^{*} G_1 \Pi_i^{*}+\Pi_i^{*} G_2 \Pi_i^{*}=(\gamma_i+\vartheta_i) \Pi_i^{*},
 $$ 
 where we have  exploited the fact that $\Pi_i^{*} G_1 \Pi_i^{*}=\gamma_i  \Pi_i^{*}$ and $\Pi_i^{*} G_2 
\Pi_i^{*}=\vartheta_i \Pi_i^{*}$. Hence, if Alice accepts $G_1,G_2$, she must also accept $G_1+G_2$ because this 
will lead to a reward of $\gamma_i+\vartheta_i$.
 Similarly, if $G$ is desirable for Alice, then also $\Pi_i^{*} (\nu G) \Pi_i^{*}=
 \nu\Pi_i^{*}  G \Pi_i^{*}$ should be desirable for any $\nu>0$. 
 
In other words, as in the case of classical desirability \citep{williams1975,walley1991}, the four conditions above state only minimal requirements: that Alice would like to increase her wealth and not decrease it (conditions $1$ and $2$); and that Alice's utility scale is linear (conditions $3$ and $4$). The first two conditions should be plainly uncontroversial. The linearity of the utility scale is routinely assumed in the theories of personal, and in particular Bayesian, probability as a way to isolate  considerations of uncertainty from those of value (wealth).

%
%

We can characterise $\domain$ also from a geometric point of view. In fact, from the above properties, it follows that a coherent set of desirable gambles $\rdesirs$ is a convex cone  in $\CH$ that includes all positive gambles (accepting partial gains) and excludes all negative gambles (avoiding partial losses). Without loss of generality we can also assume that $\rdesirs$ is not pointed, i.e., $0 \notin \rdesirs$: Alice does not accept the null gamble. A coherent set of desirable gambles is therefore a non-pointed convex cone.
 
\begin{definition}
  \label{def:sdg}
 We say that  $\domain \in \CH$ is   a {\bf coherent quantum set of desirable gambles} (DG) if
\begin{description}[noitemsep,nolistsep]
 \item[(S1)] $\domain$ is a non-pointed convex-cone (positive homogeneity and additivity);
 \item[(S2)] if $G\gneq0$ then $G \in \domain$ (accepting partial gain).
 \end{description}
 
 If in addition a coherent set of desirable gambles satisfies the following property: 
 \begin{description}[noitemsep,nolistsep]
   \item[(S3)] if $G \in \domain$ then either $G \gneq0$ or $G -\epsilon I \in  \domain$ for some strictly positive real number $\epsilon$ (openness),\footnote{In \citet{benavoli2016d} we used another formulation of openness, namely (S3'): if $G \in \domain$ then either $G \gneq0$ or $G -\Delta \in  \domain$ for some $0<\Delta \in \CH$. (S3) and (S3') are provably equivalent given (S1) and (S2).}
\end{description}
then it is said to be a coherent quantum set of {\bf strictly} desirable gambles (SDGs).
\end{definition}
\noindent Note that the although the additional openness property (S3) is not necessary for rationality, it is technically convenient as it precisely isolates the kind of models we use in QM (as well as in classical probability), see \cite{benavoli2016d}.
Property (S3) has a  gambling interpretation too: it means that we will only consider gambles that are \emph{strictly} desirable for Alice; these are the positive ones or those for which she is willing to pay a positive amount to have them.
Note that assumptions (S1) and (S2)  imply that SDGs also avoids partial loss: if $G \lneq 0$, then $G \notin \domain$
\cite[Remark III.2]{benavoli2016d}. 



In \citet{benavoli2016d}, we have shown that maximal set of SDGs, that is SDG sets that are not included in any larger set of SDG, and density matrices are one-to-one.
The mapping between them is obtained through the standard inner product in $\CH$, i.e.,
$G\cdot R= Tr(G^\dagger R)$ with $G,R\in \CH$.
This follows by a  representation result whose proof is a direct application of Hahn-Banach theorem:

\begin{theorem}[Representation theorem from \citet{benavoli2016d}]\label{thm:repr}
\label{thm:dualityopen} 
The map that associates a maximal SDG the unique density matrix $\rho$ such that 
$Tr(G^\dagger \rho) \geq  0~ \forall G \in \domain$ defines a bijective correspondence between maximal SDGs and  density matrices. 
 Its inverse is the map $(\cdot)^\circ$ that associates each density matrix $\rho$ the  maximal SDG\footnote{Here the gambles  $G\gneq0 $ are treated separately because they are always desirable and, thus, they are not informative on Alice's beliefs
about the quantum system. Alice's knowledge is determined by the gambles that are not $G\gneq0 $.}
$(\rho)^\circ=\{ G \in \CH \mid G  \gneq0\}\cup\{G \in \CH \mid  Tr(G^\dagger \rho) > 0\}$.

\end{theorem}
This representation theorem has several consequences. First, it provides a gambling interpretation of the first axiom of QM on density operators. Second, it shows that density operators are coherent, since
the set $(\rho)^\circ$ that they induce is a valid SDG. This also implies that QM is self-consistent---a gambler that uses QM to place bets on a quantum experiment cannot be made a partial (and, thus, sure) loser.  
Third, the first axiom of QM on $\mathbb{C}_h^{n \times n}$ is structurally and formally equivalent to 
Kolmogorov's first and second axioms about probabilities on $\mathbb{R}^n$ \cite[Sec. 2]{benavoli2016d}. In fact, they can be both derived via duality
from a coherent set of desirable gambles on $\mathbb{C}_h^{n \times n}$ and, respectively, $\mathbb{R}^n$.
In \citet{benavoli2016d} we have also derived Born's rule and the other three axioms of QM from the illustrated setting and shown that measurement, partial tracing and tensor product are just generalised probabilistic notions of Bayes rule, marginalisation and independence.
Finally,  as shown in \cite{BenavoliTODO}, the representation theorem enables us to derive a  Gleason-type theorem that holds for any dimension $n$ of a quantum system, hence even for $n = 2$, stating that the only logically consistent probability assignments are exactly the ones that are definable as the trace of the product of a projector and a density matrix operator.


\section{Quantum horse lotteries and preference relations}\label{sec:QLandPR}
As seen in Sec.~\ref{sec:gambling}, QM is just the Bayesian theory of probability formulated over Hermitian matrices. Now we proceed to extend such a theory of probability to make it handle non-linear utility. To this end, we work in the tradition of \citet{anscombe1963}. Central to this tradition is the notion of a horse lottery.

Consider a set of prizes $X=\{x_1,\dots,x_m\}$ with $m\geq 2$ (this last constraint will be clarified later). A \emph{horse lottery} is a compound lottery such that if $\omega\in\pspace$ occurs, it returns a simple lottery, which can depend on $\omega$, over the prizes in $X$. The idea is that at some later point the subject (Alice) will play the simple lottery thus earning one of the prizes.
Anscombe and Aumann's setting axiomatises rational preferences over horse lotteries; from this, it follows that there are probabilities and utilities that represent those preferences via maximum expected utility.

\subsection{Horse lotteries over complex numbers}
As in the classical case, now we consider that when $\Pi_i^*$ is observed, Alice receives a probability mass function $p_i$ (pmf) over the prices $X$ rather than the value $\gamma_{i}$ as in Sec.~\ref{sec:gambling}.
This framework is a composite system made of a quantum experiment and a classical experiment (on $X$). To describe it in a mathematically convenient way, we need to define gambles on this composite system.
First, we define the form of the gambles.
Since the experiment on $X$ is classical, it can be described by the subspace of $\CP$ of diagonal matrices; we denote it as $\D$.
Gambles on this composite system are therefore elements of $\D \otimes \CH\subset \CL$, where  $\otimes$ denotes the tensor product.
It can be observed that a gamble $G \in \D \otimes \CH\subset \CL$ is a block diagonal matrix with elements  in $\CH$, i.e.,
$$
G=\text{Diag}(G_1,\dots,G_m) \text{ with } G_k \in \CH.
$$
We are interested in the special case of gambles on $\D \otimes \CH$ that return a pmf $p_i$ when the quantum system is measured along some projector.
%
%
\begin{definition}\label{def:QL}
Let $Q \in \D \otimes \CH$. $Q$ is said to be a \emph{quantum horse lottery} (QH-lottery) if 
\begin{equation}\label{eq:Qlottery} 
\forall \Pi, \forall \Pi_i \in \Pi, \exists p_i \in \PO : (I_m \otimes \Pi_i) Q (I_m \otimes \Pi_i)= p_i \otimes \Pi_i, 
\end{equation}
where $\PO\subset \D$ denotes the subset of trace one diagonal-matrices whose elements are non-negative.
\end{definition}
The set $\PO$ is  isomorphic to the set of all probability mass functions (pmf) on $\reals^m$.
Therefore, with an abuse of terminology we improperly refer to the diagonal matrix $p_i$ as a pmf.
We denote the subspace of $\D \otimes \CH$ of QH-lotteries as $\QL$.
By Definition \ref{def:QL}, a QH-lottery is therefore a gamble that returns to Alice the pmf $p_i$ on $X$
whenever a measurement $\Pi$ is performed on the quantum system and the projector $\Pi_i \in \Pi$
is observed.
In what follows, we determine some properties of $\QL$.

Consider the matrix $Q$ in $\D \otimes \CH$ defined as $Q=\sum_{j=1}^n q_j \otimes V_j$, with
 $q_j \in \PO$ and $V=\{V_j\}_{j=1}^n$ is an orthogonal decomposition (OD) on $\CH$. It turns out that actually $Q$ is a QH-lottery (see Proposition \ref{prop:simpleQL} in the Appendix).  
 We call it a \emph{simple} QH-lottery.


Note that a convex combination of simple QH-lotteries is a QH-lottery. However, such a combination need not be simple anymore. 

The next theorem isolates necessary and sufficient conditions for an element of the composite space $\D \otimes \CH$ to be a QH-lottery.

\begin{theorem}\label{thm:QL}
Let $Q \in \D \otimes \CH$  of the form $Q=\text{Diag}(Q_1,\dots,Q_m)$. Then $Q \in \QL$
 if and only if  $Q_j\geq0$ 
 for every $j=1,\dots,m$ and $\sum_{j=1}^mQ_j=I_n$.
\end{theorem}

\begin{remark}
It should be observed that $Q_1,\dots,Q_m$ are  Hermitian positive semi-definite matrices  that sum up to the identity operator.
This is the definition of positive-operator valued measure (POVM). Therefore the generalisation of horse lotteries to the quantum setting naturally leads to POVMs.
\end{remark}

\begin{remark}
The classical definition of horse lotteries can be recovered by the quantum one just by considering, instead of the space $\CH$, the space of diagonal real-valued matrices. Obviously, the composite system under consideration is $\D \times \Dn$. This space is isomorphic to the space $\mathcal{L}(X \times \Omega)$ of real valued functions whose domain is $X \times \Omega$, where $\Omega =\{1, \dots, n\}$ and $X=\{1, \dots, m\}$. By applying Definition \ref{def:QL}, we immediately obtain that 
an object  $Q \in \mathcal{L} (X \times \Omega) $ 
satisfies Property \ref{eq:Qlottery}  if and only if  
$Q(\cdot,\omega)$ is a pmf on  $X$ for each $\omega \in \Omega$, meaning that $Q$ is a (classical) horse lottery.
\end{remark}


\subsection{Coherent preference relations}
Horse lotteries are given a behavioural interpretation through a notion of preference. The idea is that Alice, who aims at receiving a prize from $X$, will prefer some horse lotteries over some others, depending on her knowledge about the quantum experiment, as well as on her attitude towards the prizes.
We consider the following well-known axioms of rational preferences, formulated here in the quantum setting.

\begin{definition}
A preference relation over quantum horse lotteries is a subset $\succ \subseteq \QL \times \QL$. It is said to be {\bf coherent} if it satisfies the following axioms:

\begin{description}
\item[(A.1)] $(\forall P, Q, R \in \QL) \Big(  P \not\succ P$ and $ (P \succ Q \succ R \Rightarrow P\succ R) \Big) $ \emph{[{strict partial order}]};
\item[(A.2)] $(\forall P, Q, R \in \QL)(\forall \alpha \in (0,1]) \Big( P \succ Q \, \Leftrightarrow \, \alpha P + (1-\alpha) R \succ \alpha Q + (1-\alpha) R \Big)$ \emph{[{mixture independence}]}.
\end{description}

\end{definition}

Our approach is therefore a straightforward generalisation of the classical setting to the quantum case. 
In the classical axiomatisations of rational preferences, it is customary to assume that the preference relation has a best and a worst horse lottery \citep{luce1957games,anscombe1963}. For us it is enough to assume that the worst one exists 
and that it actually corresponds to a worst element in $X$.\footnote{The two requirements---having a worst horse lottery and a worst element in $X$---have been shown equivalent in \citet[Proposition~8]{zaffalon2015desirability}.}
Formally, we denote the last  ($m$-th) element of $X$ as $z$.  By $p_z \in \PO$ we denote the pmf that assigns all the mass to $z$, that is \[
p_z(x)= \begin{cases} 1 & \text{if } x\neq z \\ 0 & \text{else}.\end{cases}\]
Finally, by $Z$ we denote the QH-lottery $p_z \otimes  I_n$. 
Notice that $Z= \sum_{i=n}^n p_z \otimes \Pi_i$, for every OD $\Pi=\{\Pi_i\}_{i=1}^n$, and therefore
 $(I_m \otimes \Pi_i)Z(I_m \otimes \Pi_i)=p_z\otimes \Pi_i$, for every $\Pi_i \in \Pi$. 

\begin{definition}
Let $\succ \subset\QL \times \QL$ be a preference relation. We say that $\succ$ has the worst outcome if there is $z \in X$ such that $P \succ Z$ for every $P \neq Z$. 
\end{definition}

In what follows we assume that preference relations have such a worst outcome. The rationale is that the elements of $X\backslash\{z\}$ are actual prizes, whereas $z$ represents the event that no prize in $X\backslash\{z\}$ is won (nothing is won). We have assumed that $m\geq 2$ precisely as a consequence of the assumption that there is the worst-outcome $z$ among the elements of $X$.

The scaled differences of QL-lotteries is the set defined by
\begin{equation}\label{eq:setA}
\mathcal{A}=\{ \lambda(P-Q) \mid \lambda >0, P,Q \in \QL\},
\end{equation}
where $\lambda$ is a positive real. 
The set $\mathcal{A}$ constitutes a vector space (Proposition \ref{prop:vectspace} in the Appendix).
\begin{theorem}
\label{thm:C-convex-cone}
The map
\begin{equation}\label{eq:C-convex-cone} \succ \; \mapsto \; \cone =\{\lambda(P-Q) \mid  P,Q \in \QL, ~\lambda >0, P\succ Q\}\end{equation}
determines a bijection between non-empty coherent preference relations over $\QL$ and non-empty non-pointed convex cones in $\mathcal{A}$.
\end{theorem}
Thus, it turns out that non-pointed cones and coherent quantum preference relations are just two ways of looking at the same thing.


\section{Quantum desirability vs quantum preference: two faces of the same coin}\label{sec:desirablegambles}

In this section we follow the same strategy as in \cite{zaffalon2015desirability} to establish an equivalence between the theories of coherent quantum preference and coherent quantum desirability. To this end, we first define the projection operator that drops the $z$-components from an act.
\begin{definition}  The \emph{projection operator} is the functional
$\mathsf{Proj}:\D \otimes \CH \rightarrow \DP \otimes \CH$ that takes the QL-lottery ($m$-block diagonal matrix) $Q$ and returns
$\mathsf{Proj}(Q)=\text{Diag}(Q_1,\dots,Q_{m-1})$.
\end{definition}
In this paper, we are going to use this operator to project QH-lotteries in $\D \otimes \CH$into gambles on
$\DP \otimes \CH$. 
However, instead of working directly with the space $\QL$, in what follows it will be more convenient to reason on the space $\mathcal{A}$ of scaled differences of QH-lotteries defined in~\eqref{eq:setA}. 
Note also that the restriction of $\mathsf{Proj}$ to
$\mathcal{A}$
 is  injective. 

Based on the correspondence between cones on $\mathcal{A}$ and preference relations, it is then an easy step to show (see Proposition \ref{prop:rdesirs} in the Appendix) that given a coherent preference relation $\succ$, one can define a coherent set $\rdesirs$ of desirable gambles on $\DP \otimes \CH$ as
$\rdesirs = \{\lambda\, \mathsf{Proj}(P-Q):P\succ
Q,\lambda>0\}\label{eq:RfromC}
$ 
and with the property that 
\begin{equation}\label{eq:order}
P \succ Q\iff\mathsf{Proj}(P-Q)\in\rdesirs.
\end{equation}

One can actually verify that there is an exact correspondence between coherent sets of desirable gambles and coherent preference relations.

\begin{theorem}\label{thm:1to1}
There is a one-to-one correspondence between coherent sets of desirable gambles over $\DP \otimes \CH$  and coherent preference relations over $\QL\times\QL$. 
\end{theorem}



\section{Archimedeanity and the representation  theorem}\label{sec:objective}

Archimedeanity is an extra axiom adopted in traditional axiomatisations of rational preferences that tries to capture a form of continuity; it is such an axiom that makes it possible to have a representation of preferences in terms of expected utilities. \citet[Prop. 11]{zaffalon2015desirability} (and in the quantum case, we, by Proposition \ref{prop:archilimit} in the Appendix) have shown that the traditional Archimedean axiom has some drawbacks that can be fixed with a slight change in its definition. It is based on the notion of objective preference.
\begin{definition}[Objective preference]
Let $P,Q \in \QL$. We say that \emph{$P$ is objectively preferred to $Q$} if $\Proj(P-Q)\gneq 0$. 
We denote objective preference by $P \rhd Q$.
\end{definition}
(Note that the definition neglects the outcome $z$, since it is not one any subject actually wants.)

Objective preference is a preference relation. Moreover, it is the least preference relation over $\QL \times \QL$ in the sense that it is included in any other preference relation (in this sense, we call it ``objective''). Now we can directly rephrase Zaffalon and Miranda's Archimedean notion as follows for the quantum case:
 \begin{description}
\item[(A.3)] $(\forall P, Q \in \QL)\ P \succ Q, P \not\vartriangleright Q \Rightarrow (\exists\alpha \in (0,1))\  \alpha P + (1-\alpha) Z \succ Q$ [\emph{Weak Archimedeanity}].
\end{description} 
Analogously to their case, we obtain that it is equivalent to use coherent quantum sets of strictly desirable gambles in order to represent weakly Archimedean coherent preference relation on quantum horse lotteries. Recall also that a preference relation $\succ$ is said to be complete (or total) if either $P \succ Q$ or $Q \succ P$, for every $P,Q \in \QL$ with $P\neq Q$.

\begin{theorem}\label{thm:weak1to1}
There is a one to one correspondence between coherent sets of SDG over $\DP \otimes \CH$  and coherent preference relations over $\QL\times\QL$ that are weakly Archimedean. Moreover, such a correspondence induces a bijection between maximal coherent sets of SDG  and complete weakly Archimedean coherent preference relations.
\end{theorem}


Based on Theorem \ref{thm:weak1to1}, we can then obtain a representation theorem for complete weakly Archi\-me\-dean coherent preference relations along these lines.
First of all, Theorem \ref{thm:repr} from \citet{benavoli2016d} can be restated  in the case of quantum horse lotteries as follows.
\begin{theorem}\label{thm:duality}
The map that associates a maximal SDG over $\DP \otimes \CH$ the unique trace-one positive matrix $R\in \DP \otimes \CH$ such that 
$Tr(G^\dagger R) \geq  0~ \forall G \in \domain$ defines a bijective correspondence between maximal SDGs over $\DP \otimes \CH$ and  trace-one positive matrices over $\DP \otimes \CH$. 
 Its inverse is the map $(\cdot)^\circ$ that associates each trace-one positive matrix $R$ the  maximal SDG
\begin{equation}
\label{eq:induced}
(R)^\circ=\{ \DP \otimes \CH \mid G  \gneq0\}\cup\{G \in \DP \otimes \CH \mid  Tr(G^\dagger R) > 0\}. 
\end{equation}
\end{theorem}

All trace-one positive  matrices $R$ are of the form $R=\diag \left( p_1 \rho_1, \dots, p_{m-1} \rho_{m-1} \right)$ with $\rho_i \in \CH$ being density matrices and $\diag(p_1, \dots, p_{m-1} )\in \mathbb{P}^{m-1}$.
Hence, applying Theorems \ref{thm:weak1to1} and \ref{thm:duality}  to Property \ref{eq:order} yields the following representation result for complete preference relations:

\begin{corollary}\label{cor:duality}
A  relation $\succ$ over $\QL\times \QL$ is a complete weakly Archimedean coherent preference relation if and only if there is a unique trace-one positive matrix $R=\diag \left( p_1 \rho_1, \dots, p_{m-1}\rho_{m-1} \right)$ such that
\begin{equation}
P \succ Q \Leftrightarrow \Big( \text{either } P \rhd Q \text{ or }\sum_{i=1}^{m-1}p_i  Tr(P_i^\dagger \rho_i) >   \sum_{i=1}^{m-1}p_i  Tr(Q_i^\dagger \rho_i) \Big) ~ \forall P, Q \in \QL.
\end{equation} 
\end{corollary}
Consistently with our generalisation of Gleason's theorem \citep{BenavoliTODO}, this result holds in any dimension (even $n=2$), because we ask  preference relations to be coherent.


\section{Coherent updating and state independence}\label{sec:update}

This section shows how to derive in  a very simple, elegant, way a coherent rule for updating preferences. In particular our aim is to answer this question: how should Alice change her preferences in the prospect of obtaining new information in the form of an event?

We initially assume that Alice considers an event ``indicated'' by a certain projector $\Pi_i \in \CH$
in $\Pi=\{\Pi_i\}_{i=1}^n$. The information it represents is: an experiment $\Pi$ is performed and the event indicated by $\Pi_i$ happens.\footnote{{
We assume  that the quantum measurement device is a ``perfect meter'' (an ideal common  assumption in QM), i.e., there are not observational errors---Alice can trust the
received information.}} 

Now, assume that Alice's preferences are modelled by the coherent relation $\succ$ on $\QL$. From Theorem~\ref{thm:1to1} we can consider the coherent set $\rdesirs$ in $\DP \otimes \CH$. 
Hence, we reason as in the derivation of the second axiom of QM in \citet[Sec. V]{benavoli2016d}.
Under the assumption that an experiment $\Pi$ is performed and the event indicated by $\Pi_i$ happens, Alice can focus on gambles that are contingent on $I_{m-1}\otimes\Pi_i$: these are the gambles such that ``outside'' $I_{m-1}\otimes\Pi_i$ no utile is received or due---status quo is maintained---; in other words, they represent gambles that are called off if the outcome of the experiment is not $\Pi_i$. 
Mathematically, these  gambles are of the form
$$
G=\left\{\begin{array}{ll}
 H &  \text{if } I_{m-1}\otimes\Pi_i \text{ occurs},\\
 0 &  \text{if } I_{m-1}\otimes\Pi_j \text{ occurs, with}  ~j\neq i.\\
  \end{array}\right.
$$

In this light, we can define Alice's conditional preferences by moving to the equivalent view on gambles, restricting the attention to gambles of the form $(I_{m-1}\otimes\Pi_i) G (I_{m-1}\otimes\Pi_i)=p\otimes\Pi_i$ with $G \in \DP \otimes\CH$, and finally updating the preferences by looking at the corresponding preference relation.

\begin{definition}
 Let  $\succ$ be a  preference relation. The relation obtained as $\succ_{\Pi_i}:=\Proj_1^{-1}(\rdesirs_{\Pi_i})$, with
\begin{equation}
\label{eq:condition}
\rdesirs_{\Pi_i}=\left\{G \in \DP \otimes\CH \mid  G \gneq0 \textit{ or }(I_{m-1}\otimes\Pi_i) G (I_{m-1}\otimes\Pi_i) \in \rdesirs \right\}
\end{equation} 
is  called the {\bf preference relation conditional} on  $\Pi_i$.
\end{definition}

By the same argument as in \citet[Prop. A.6]{benavoli2016d}, one can prove that $\rdesirs_{\Pi_i}$ is a coherent set of (strictly) desirable gambles, whenever $\rdesirs$ is a coherent set of (strictly) desirable gambles. By Theorems \ref{thm:1to1} and \ref{thm:weak1to1} this yields that:

\begin{theorem}\label{thm:update}
Let $\succ$ be a (weakly Archimedean) coherent preference relation.  The relation $\succ_{\Pi_i}$ conditional on the event $\Pi_i$ is also a (weakly Archimedean) coherent preference relation.
\end{theorem}

Now, we rely on conditioning to introduce the concept of state-independent preferences. For this purpose, we use results in
\cite{benavoli2016d} to prove the fourth postulate of QM about composite systems. 
We first  define the concept of epistemic irrelevance.
\begin{definition}\label{def:epirrelev}
Let $\mathcal{R}\subset  \DP \otimes \CH,$ and let us define
$$
\begin{array}{rcl}
\marg_{\DP}(\mathcal{R})&=&\left\{G \in \DP \mid G \otimes I_n \in \mathcal{R}\right\},\\
\end{array}
$$
An SDG $\rdesirs$ on $\DP \otimes \CH$  is said to satisfy {\bf epistemic irrelevance}
of the subsystems $\DP$ to $\CH$ when $\marg_{\DP}(\rdesirs)=\marg_{\DP}(\rdesirs_{\Pi_i})$ for each projection measurement $\Pi=\{\Pi_i\}_{i=1}^n$.
 \end{definition}
 Let us briefly explain this definition. The set $\rdesirs_{\Pi_i}$ is the SDG conditional on the event indicated 
by 
 $\Pi_i$, as it follows from its definition and~\eqref{eq:condition}. Thus,  $\marg_{\DP}(\rdesirs)=\marg_{\DP}(\rdesirs_{\Pi_i})$  means that
Alice's marginal SDG $\marg_{\DP}(\rdesirs)$ on the subsystem $\DP$ and the marginal on $\DP$ of Alice's
SDG updated with the information ``the event indicated by $\Pi_i$ has happened'', which is $\marg_{\DP}(\rdesirs_{\Pi_i})$, coincide.
If this holds for all possible $\Pi_i$'s, then any information on $\CH$ does not change Alice's beliefs on $\DP$: this is precisely
the definition of epistemic irrelevance.  In case $\rdesirs$ is maximal and satisfies  epistemic irrelevance, we have shown in \citet[Sec. VII.c]{benavoli2016d} that the representation Theorem \ref{thm:duality} applied to such $\rdesirs$ defines a matrix $R$ that factorizes as  $R=p\otimes \rho$. Therefore, as in the classical framework for decision theory, the ``joint'' density matrix $R$ factorizes as the product of  $p \in \mathbb{P}^{m-1}$ and the density matrix  $\rho \in \CH$. Stated otherwise,  $\rdesirs$ models independence  between utility ($p$) and the ``probabilistic'' information ($\rho$) associated to the quantum system.  Alice's preferences are  thusly state-independent. 

In case of epistemic irrelevance, 
Corollary \ref{cor:duality} can therefore be reformulated as follows.

\begin{corollary}\label{cor:dualityirrelevance}
Let $\succ$ be  a complete weakly Archimedean coherent preference relation over $\QL\times \QL$ satisfying epistemic irrelevance, then there is a unique trace-one positive matrix $R=p\otimes \rho$ with $p=\diag \left( p_1, \dots, p_{m-1} \right)  \in \mathbb{P}^{m-1}$ and   $\rho \in \CH$ such that
\begin{equation}
P \succ Q \Leftrightarrow \Bigg( \text{either } P \rhd Q \text{ or }  Tr\left(\left(\sum_{i=1}^{m-1}p_i P_i^\dagger\right) \rho\right) >   Tr\left(\left(\sum_{i=1}^{m-1}p_i Q_i^\dagger\right) \rho\right) ~~\Bigg)~ \forall P, Q \in \QL.
\end{equation} 
\end{corollary}
The inner term $\sum_{i=1}^{m-1}p_i P_i^\dagger$ (resp.\ $\sum_{i=1}^{m-1}p_i Q_i^\dagger$) can be interpreted as the utility of $P$ (resp.\ $Q$) and, therefore, we can also rewrite the above inequality as $Tr(u(P)\rho)>Tr(u(Q)\rho)$.  On the other hand, the trace inner product is the usual way in quantum mechanics of computing the expectation with respect to the density matrix $\rho$.

\section{Related work}
Axiomatic frameworks for the theory of subjective expected utility were originally given by \cite{savage1954foundations} and by~\cite{anscombe1963}. 
\cite{karni2013axiomatic} provides a  recent overview of several  variations and extensions of these two models. 

\citeauthor{busemeyer2012quantum}'s~\citeyearpar{busemeyer2012quantum} book presents an overview of quantum-like approaches to cognition and decision theory.  \cite{deutsch1999quantum}, \cite{khrennikov2016quantum} and \cite{danilov:halshs-01324046} are examples of other works addressing  similar issues. 
In particular the latter proposes an axiomatisation for quantum preferences directly in the space of Hermitian matrices similar to the one presented here. However, the authors only consider what we call \emph{simple} quantum horse lotteries. In doing so, the traditional axiom of mixture independence is formulated relative to the particular orthogonal decomposition associated to a simple lottery; an additional axiom becomes then necessary to bind lotteries based on different orthogonal decompositions. Moreover their representation theorem crucially employs (original) Gleason's theorem and therefore only works on spaces of dimension at least three. Because of those characteristics, it is unclear to us whether or not that axiomatisation is coherent: e.g., whether it guarantees that a subject whose quantum preferences on a space of dimension two cannot be made a sure loser, that is, shown to be irrational. The case of dimension two is particularly critical as dispersion-free probabilities---which \cite{BenavoliTODO} have shown to be incoherent---could in principle be employed.


\section{Concluding remarks}
In this paper, we have axiomatised rational preferences over quantum horse lotteries. Such a development is a natural follow up of our recent work \citep{benavoli2016d}, which has shown that Quantum  Mechanics is the Bayesian theory of probability over Hermitian matrices. By bridging those rational preferences with quantum desirability, we have given a representation theorem in terms of quantum probabilities and utilities. 


There are many directions that can be explored starting from in this paper. Two of them are particularly  important in our view. The first regards the full extension of our setting to partial (i.e., incomplete) preferences; this would enable it to deal with sets of quantum probabilities and utilities. Our axiomatisation is actually already conceived for incomplete preferences. This should be clear from the equivalence between preference and desirability. However, giving a full account of partial preferences requires detailing a few questions. The second is the definition of horse lotteries in their full generality as compound quantum lotteries: that is, such that the simple lotteries (over prizes) they embed be quantum too rather than classical. We leave these topics for future work.




\appendix

\section{Propositions and proofs}\label{sec:app}

 \subsection{Results of Section \ref{sec:QLandPR}}

 \begin{proposition}\label{prop:simpleQL}
 Let us consider the matrix $Q$ in $\D \otimes \CH$ defined as $Q=\sum_{j=1}^n q_j \otimes V_j$, with
 $q_j \in \PO$ and $V=\{V_j\}_{j=1}^n$ is an orthogonal decomposition (OD) on $\CH$, then $Q\in \QL$. 
 We call it a \emph{simple QH-lottery}.
\end{proposition}
\begin{proof}
 The proof is immediate.  From the definition
 of QH-lottery
 \begin{align}
 \nonumber
  (I_m \otimes \Pi_i) Q (I_m \otimes \Pi_i)&=\sum_{j=1}^n q_j \otimes (\Pi_iV_j\Pi_i)=\sum_{j=1}^n q_j \otimes \gamma_j \Pi_i=\sum_{j=1}^n \gamma_j q_j \otimes \Pi_i=p_i \otimes \Pi_i,
 \end{align}
 with $p_i \in \PO $, where last equality follows by $\gamma_j=\pi_i^\dagger v_jv_j^\dagger \pi_i\geq0$ and $\sum_{j=1}^n \gamma_j=\pi_i^\dagger (\sum_{j=1}^n  v_jv_j^\dagger) \pi_i=\pi_i^\dagger I_n \pi_i=1$ and we have exploited that  $\Pi_i=\pi_i\pi_i^\dagger$ and $V_j=v_jv_j^\dagger$
 for $\pi_i,v_j\in \mathbb{C}^n$.
\end{proof}

\begin{proof}[Proof of Theorem \ref{thm:QL}]
First notice that we can rewrite $Q$ as  $ Q=\sum_{k=1}^m  e_ke^T_k\otimes Q_k$
 with $Q_k\in \CH$ and $e_k$ is the canonical basis for $\reals^m$. 
 Assume $Q \in \QL$.
 Consider an arbitrary OD $\Pi$, and $\Pi_i \in \Pi$. Hence, we have that
  \begin{align}
 \label{eq:hlotteryM}
  (I_m \otimes \Pi_i) Q (I_m \otimes \Pi_i)&= \sum_{k=1}^m   (I_m \otimes \Pi_i)(e_ke^T_k\otimes Q_k)(I_m \otimes \Pi_i)= \sum_{k=1}^m   e_ke^T_k\otimes \Pi_iQ_k\Pi_i.
 \end{align}
 To satisfy Property \eqref{eq:Qlottery} in the definition of QH-lotteries,  we want last term to be equal of $p_i \otimes \Pi_i$ for some $p_i \in \PO$.
 This implies that: (1) $\Pi_i Q_k \Pi_i\geq0 $; (2) 
  $\sum_{k=1}^m \Pi_iQ_k\Pi_i=\sum_{k=1}^m p_i(k)\Pi_i=1\Pi_i$
  Given the arbitrariness of  $\Pi$, this means that $Q_k$ is positive semi-definite
  and  $\sum_{k=1}^m Q_k=I_n$. In fact assume that $Q_k$ has some negative eigenvalue with corresponding
  projector $V$. Then take $\Pi_i=V$ so that $\Pi_i Q_k \Pi_i<0$ a contradiction.
  Conversely assume that  $Q_k\geq0$ for every $k$ and $\Pi_iQ_k\Pi_i=p_i(k)\leq1$ for $k=1,\dots,m-1$
  and $\sum_{k=1}^{m-1}p_i(k)\leq 1$ (otherwise it could not return a probability), then 
  $\Pi_iQ_m\Pi_i=(1-\sum_{k=1}^{m-1}p_i(k))\Pi_i=\Pi_i(I- \sum_{k=1}^{m-1}Q_k)\Pi_i$.
  That implies that $\Pi_i(I- \sum_{k=1}^{m}Q_k)\Pi_i=0$ that ends the proof of one direction of the Theorem.
   The other direction is trivial.  
\end{proof}

\begin{proposition}\label{prop:vectspace}
The set $\mathcal{A}$ constitutes a vector space. 
\end{proposition}
\begin{proof}
The only thing to verify is closure under addition. Hence let $\lambda(P-Q), \mu(R-S) \in \mathcal{A}$. Let $\alpha=\frac{\lambda}{\lambda + \mu}$ and $\beta =\frac{\lambda}{\alpha}$ . Notice that $\alpha P +(1-\alpha) R, \alpha Q + (1-\alpha) S \in \QL$. Hence 
$\beta((\alpha P +(1-\alpha) R)  - (\alpha Q + (1-\alpha) S))= \lambda(P-Q) + \mu(R-S) \in \mathcal{A}$.
\end{proof}

\begin{proof}[Proof of Theorem \ref{thm:C-convex-cone}]
We show that the map
\begin{equation} \succ \; \mapsto \; \cone =\{\lambda(P-Q) \mid  P,Q \in \QL, ~\lambda >0, P\succ Q\}\end{equation}
determines a bijection between non-empty coherent preference relations over $\QL$ and non-empty non-pointed convex cones in $\mathcal{A}$.

Let $\cone$ be a non-pointed convex cone in $\mathcal{A}$, and define $P \succ Q$ if and only if $\lambda(P-Q) \in \cone$, for some $\lambda>0$. Since $\cone$ is a cone, $P \succ Q$ if and only if $\lambda(P-Q) \in \cone$, for every $\lambda>0$. From non-pointedness of $\cone$, $\succ$ is irreflexive. Transitivity is immediate. Mixture independence follows form the fact that $\alpha(P-Q)=(\alpha P + (1-\alpha) R) - (\alpha Q + (1-\alpha) R)$.
For the other direction, the fact that $\cone$ is non-pointed follows from  non reflexivity of $\succ$. For the other properties, scaling is immediate. For addition, 
let $\lambda(P-Q), \mu(R-S) \in \cone$. This means that $P \succ Q$ and $R \succ S$. Define  $\alpha$ and $\beta$ as in Proposition \ref{prop:vectspace}. By mixture independence $ \alpha P + (1-\alpha R) \succ \alpha Q + (1-\alpha R) $ and $ \alpha Q + (1-\alpha R) \succ \alpha Q + (1-\alpha S) $. By transitivity $ \alpha P + (1-\alpha R) \succ  \alpha Q + (1-\alpha S) $, and therefore $ (\alpha P + (1-\alpha R)) -  (\alpha Q + (1-\alpha S)) \in \cone$. By scaling $\beta((\alpha P +(1-\alpha) R)  - (\alpha Q + (1-\alpha) S))= \lambda(P-Q) + \mu(R-S) \in \cone$.
\end{proof}

\subsection{Results of Section \ref{sec:desirablegambles}}

\begin{proposition}\label{prop:rdesirs}
Let $\succ$ be a coherent preference relation on $\QL\times\QL$ and $\rdesirs$ be defined by
\begin{equation}
\rdesirs = \{\lambda\, \mathsf{Proj}(P-Q):P\succ
Q,\lambda>0\}.\label{eq:RfromC0}
\end{equation}
Then:
\begin{itemize}
\item[(i)] $\rdesirs$ is a coherent set of desirable gambles on $\DP \otimes \CH$.
\item[(ii)] $P \succ Q\iff\mathsf{Proj}(P-Q)\in\rdesirs.$
\end{itemize}
\end{proposition}
\begin{proof}
The proof of the second point is trivial, so we concentrate on showing that $\rdesirs$ satisfies the properties of Definition \ref{def:sdg}.
We consider the set $\mathcal{A}$ of scaled differences of QL-lotteries.
Hence, $\rdesirs=\Proj(\cone)$, and we can use the fact that $\cone$ is a non-pointed convex cone to assess the desired properties of $\rdesirs$.

\begin{itemize}
\item (accepting partial gain) 
Consider a PSDNZ block-diagonal matrix $G\in \DP\times\CH$ and  let $\lambda= \sum_j \max \mathsf{EigenValue}(G_j)$. Note that $\lambda>0$ otherwise $G$ would not be PSDNZ. Then let 
\begin{equation*}
P=\text{Diag}\left(G_1/\lambda,\dots,G_{m-1}/\lambda,I_n-\sum_{j} G_j/\lambda\right),
\end{equation*}
and  $I_n-\sum_{j} G_j/\lambda\geq0$ by  Weyl's inequality. Hence, $P \in \QL$ by Theorem \ref{thm:QL}.
It follows that  $\mathsf{Proj}(\lambda(P-Z))=G$ and since $P\neq Z$ because
$G$ is (PSD)NZ, then $P\succ Z$ and hence $G\in\rdesirs$.

\item Positive homogeneity follows trivially from the convexity of $\cone$.

\item That $\rdesirs$ is a cone follows from the fact that $\cone$ is a cone and by taking into consideration that $\mathsf{Proj}$ is a linear
functional. Moreover  $\cone$ does not contain the origin and hence  $0\notin\rdesirs$.
\end{itemize}
\end{proof}

\begin{proof}[Proof of Theorem \ref{thm:1to1}]
As before, we show instead that there is a one to one correspondence between coherent sets of desirable gambles and non-pointed convex cones on $\mathcal{A}$.
We already know that $\Proj_2$ is injective. We thence verify that $\Proj_2$ is also surjective. Let $W=\text{Diag}\left(W_1,\dots,W_{m-1}\right) \in \DP \otimes \CH$. For each $j<m$, consider $W_j=\sum_{i=1}^{n} \gamma_i \Pi_i$. Write any eigenvalue $\gamma_i$ of $W_j$ as $\gamma_i=\gamma^+_i-\gamma^-_i$
where $\gamma^+_i=\max(0,\gamma_i)$ and $\gamma^-_i=\max(0,-\gamma_i)$.
Then we have that 
$W_j=\sum_{i=1}^{n} \gamma^+_i \Pi_i-\sum_{i=1}^{n} \gamma^-_i \Pi_i=W^+_j-W^-_j$.
Let $\lambda^+=(m-1)\max_j(Eigen(W^+_j))$ and $\lambda^-=(m-1)\max_j(Eigen(W^-_j))$. Define
$\lambda=\lambda^++\lambda^- > 0$. Consider
\[P=\text{Diag}\left(W^+_1/\lambda,\dots,W^+_{m-1}/\lambda,I-\sum_{j}W^+_{j}/\lambda\right)\]  and 
\[Q=\text{Diag}\left(W^-_1/\lambda,\dots,W^-_{m-1}/\lambda,I-\sum_{j}W^-_{j}/\lambda\right).\]
By construction, $W^+_\ell/\lambda,W^-_\ell/\lambda\geq0$, for $\ell<m$. By Weyl's inequality, $I-\sum_{j}W^+_{j}/\lambda, I-\sum_{j}W^-_{j}/\lambda\geq0$. Hence,  $P,Q \in \QL$ by Theorem \ref{thm:QL}. We conclude by noticing that $W=\mathsf{Proj}(\lambda(P-Q))$.
\end{proof}

\subsection{Results of Section \ref{sec:objective}}

 If a coherent preference relation $\succ$ also satisfies the next axiom, we say that it is Archimedean.
 \begin{description}
\item[(fA.3)] $(\forall P, Q, R \in \QL) P \succ Q \succ R \Rightarrow (\exists \alpha, \beta \in (0,1)) : \alpha P + (1-\alpha) R \succ Q \succ \beta P + (1-\beta) R$ [\emph{Archimedeanity}]
\end{description} 

\begin{proposition}\label{prop:archilimit}
An objective preference relation $\rhd$ is Archimedean if and only if $m=2$ and $n=1$.
\end{proposition}
\begin{proof}
We begin with the direct implication. Assume first of all that
$m\geq3$.  Fix a OD $\Pi$, and let $U=\sum_{i=1}^n u\otimes \Pi_i$, where $u(\ell) = \frac{1}{m}$. 
Then
fix any $k\neq m$ and $\eps>0$ small enough so as to define lottery $P=\sum_{i=1}^np\otimes \Pi_i$
where
\begin{equation*}
p(\ell)=
\begin{cases}
u(\ell)+\eps\quad&\text{ if }\ell=k,\\
u(\ell)-\eps\quad&\text{ if }\ell=m,\\
u(\ell)\quad&\text{ otherwise.}
\end{cases}
\end{equation*}
It follows that $P\succ U$ since $P\rhd U$. Given that also $U\succ
Z$, we can apply the Archimedean axiom to obtain that there
is a $\beta\in(0,1)$ such that:
\begin{equation*}
\beta P + (1-\beta)Z \succ U.
\end{equation*}
But $\beta P+ (1-\beta)Z \not\vartriangleright U$. Indeed, for any $\ell\neq k, m$, it holds that $\beta p(\ell)
+ (1-\beta)p_z(\ell) - u(\ell)=(\beta-1)\frac{1}{m}<0$.

Next, if $m=2$ and $n\geq 2$, and fix an OD $\Pi$, and $\Pi_1 \in \Pi$. Define QL-lotteries $P= p \otimes \Pi_1 + \sum_{i=2}^n r_p \otimes \Pi_i$  and $Q= q \otimes \Pi_1 + \sum_{i=2}^n r_q \otimes \Pi_i$ where
$p(1)=q(2)=1, q(1)=p(2)=0, r_p(\ell)=r_q(\ell)=\frac{1}{2}$. Then $P\rhd Q \rhd Z$, whence $P\succ Q
\succ Z$. However, for any $\beta\in (0,1)$ it holds that
\begin{equation*}
 \beta r_p(1)+ (1-\beta) p_z(1) < r_q(1),
\end{equation*}
whence $\beta P +(1-\beta) Z\not\vartriangleright Q$ and as a consequence $\succ$
is not Archimedean.

Conversely, given $m=2$ and $n=1$, then it holds
that $P\rhd Q \iff p(1)>q(1)$, and then if $P\succ Q
\succ R$ there are $\alpha,\beta\in (0,1)$ such that $\alpha
p(1)+ (1-\alpha) r(1) > q(1) > \beta
p(1)+(1-\beta) q(1)$, meaning that $\alpha
P+(1-\alpha) R \succ Q \succ \beta P+(1-\beta) R$, and as a
consequence $\succ$ is Archimedean.
\end{proof}

By $\gambles^+$ we denote the set of all gambles $G$ in $ \DP \otimes \CH$ such that $G \gneq 0$.

\begin{lemma}\label{lem:trans}
Let $\rdesirs$ be the coherent set of desirable gambles on $\DP \otimes \CH$ arising from a coherent preference relation $\succ$. Denote by $U$ the uniform QL-lottery defined by $U=\frac{1}{m} \text{Diag}(I_{n},\dots,I_{n})$. Consider the following translated sets:
\begin{eqnarray*}
\tau(\rdesirs)&=&\{G\in\DP \otimes \CH \mid G=\mathsf{Proj}(U)+F, F\in\rdesirs\},\\
\tau(\gambles^+
)&=&\{G\in\DP \otimes \CH \mid G-\mathsf{Proj}(U)\gneq 0\},
\end{eqnarray*}
and the following version of (S3-openness) adapted to $\tau(\rdesirs)$:
\begin{description}
\item[(S3')] $g\in \tau(\rdesirs)\setminus \tau(\gambles^+)\Rightarrow(\exists\delta>0)G-\delta I_{n(m-1)}\in \tau(\rdesirs)$.
\end{description}
Then it follows that:
\begin{itemize}
\item[(i)] (S3) and (S3') are equivalent conditions.
\item[(ii)] $\tau(\rdesirs)=\{\mathsf{Proj}(U)+\lambda\mathsf{Proj}(P-U):\lambda>0,P\succ U\}$.
\end{itemize}
\end{lemma}
\begin{proof}
Note first that if $G=\mathsf{Proj}(U)+F$, $G\in \tau(\rdesirs)$, then $F\in\rdesirs$: in fact if $G\in \tau(\rdesirs)$, then $G=\mathsf{Proj}(U)+H$ for some $H\in\rdesirs$, whence $F=H$.

\begin{itemize}

\item[(i)] Let us show that (S3) and (S3') are equivalent conditions.

Consider $G\in \tau(\rdesirs)\setminus
\tau(\gambles^+)$; then $G=\mathsf{Proj}(U)+F$,
$F\in\rdesirs\setminus\gambles^+$.
Applying (S3), there is $\delta>0$ such that
$F-\delta I_{n(m-1)}\in\rdesirs$, whence $\mathsf{Proj}(U)+F-\delta I_{n(m-1)}=G-\delta I_{n(m-1)}\in
\tau(\rdesirs)$.

Conversely, if $F\in\rdesirs\setminus\gambles^+$, then $\mathsf{Proj}(U)+F\in \tau(\rdesirs)\setminus \tau(\gambles^+)$; applying (S3') we get that there is $\delta>0$ such that $\mathsf{Proj}(U)+F-\delta I_{n(m-1)} \in \tau(\rdesirs)$, whence $F-\delta I_{n(m-1)}\in\rdesirs$.

\item[(ii)] 
It is trivial that $\{\mathsf{Proj}(U)+\lambda\mathsf{Proj}(P-U):\lambda>0,P\succ
U\}\subseteq \tau(\rdesirs)$, so we concentrate on the converse
inclusion.

Consider $\mathsf{Proj}(U)+F$. Then there are $\lambda>0$ and $P_1\succ P_2$ such that $F=\lambda\mathsf{Proj}(P_1-P_2)$. Remember that $m\geq2$ by assumption; whence, if we take $\mu\in(0,\frac{1}{\lambda m})$, we obtain that $U+\mu\lambda(P_1-P_2)\geq0$. Let $P_i= \text{Diag}(P^i_1,\dots,P^i_m)$, $i=1,2$. Notice that $P_1 - P_2= \text{Diag}(P^1_1-P^2_1,\dots,P^1_m - P^2_m)$ and therefore $\sum_{j=1}^m(P^1_j - P^2_j)=0$. Hence 
$\sum_{j=1}^m[\frac{1}{m}I_n+\mu\lambda(P_1-P_2)]=I_n$. We deduce that $P= U+\mu\lambda(P_1-P_2)\in\QL$.

Given that $P-U=\mu\lambda(P_1-P_2)$ and that $P_1\succ P_2$, by coherence of $\succ$ it is not difficult to verify that $P\succ U$. And since
$\mathsf{Proj}(U)+F=\mathsf{Proj}(U)+\frac{1}{\mu}\mathsf{Proj}(P-U)$, we deduce that
$\mathsf{Proj}(U)+F\in\{\mathsf{Proj}(U)+\lambda\mathsf{Proj}(P-U):\lambda>0,P\succ U\}$.
\end{itemize}
\end{proof}

\begin{proposition}\label{prop:lifeisgood}
$\rdesirs$ is a coherent set of strictly desirable gambles if an only if
$\succ$ is a coherent weakly Archimedean preference relation 
\end{proposition}

\begin{proof}
For the direction from left to right, we reason as follows.
Coherence follows from Proposition~\ref{prop:rdesirs}. In order to
prove that $\rdesirs$ is a set of strictly desirable gambles, it is
enough to show that (S3') holds, thanks to
Lemma~\ref{lem:trans}.

Consider $G\in \tau(\rdesirs)\setminus
\tau(\gambles^+)$. Then, according to Lemma~\ref{lem:trans}, $G=\mathsf{Proj}(U)+F$, with $F=\lambda\mathsf{Proj}(P-U)\in\rdesirs$, $\lambda>0$, and for some $P\in\QL$ such that $P\succ U, P\not\vartriangleright U$.

Since $U\succ Z$, it follows from (A3) that there is some
$\beta\in(0,1)$ such that $\beta P+(1-\beta)Z\succ U$, whence
$$\mathsf{Proj}(\beta (P-U)-(1-\beta)(U-Z))=\frac{\beta F}{\lambda}-(1-\beta)\mathsf{Proj}(U)\in\rdesirs.$$ Since $\rdesirs$ is a cone, this
means that $F-\frac{\lambda(1-\beta)}{\beta} \mathsf{Proj}(U) \in\rdesirs$ for
some $\beta\in (0,1)$, or, in other words, that $F-\delta I_{n(m-1)} \in\rdesirs$, with $\delta=\frac{\lambda(1-\beta)}{\beta m}>0$. We deduce that $G-\delta I_{n(m-1)}=\mathsf{Proj}(U)+F-\delta I_{n(m-1)}\in \tau(\rdesirs)$, whence (S3') holds and as a consequence $\rdesirs$ is a coherent
set of strictly desirable gambles.

Finally we verify that if $\rdesirs$ is a coherent set of strictly desirable gambles, then
$\succ$ is weakly Archimedean.
 Let us consider $P\succ Q,P\not\vartriangleright Q$.
By (S3), there is $\delta>0$ such that
$\mathsf{Proj}(P-Q)-\delta I_{n(m-1)}\in\rdesirs$. Choose $\beta\in(0,1)$ so that
$\mathsf{Proj}((1-\beta)P)-\delta I_{n(m-1)}\leq 0$. Then $\mathsf{Proj}(\beta P -
Q)- (\mathsf{Proj}(P-Q)-\delta I_{n(m-1)}) \geq 0$, whence $\mathsf{Proj}(\beta P-Q)\in\rdesirs$, which
implies that $\beta P + (1-\beta)Z\succ Q$. 
 \end{proof}
 
 \begin{proof}[Proof of Theorem \ref{thm:weak1to1}]
 The first point is an immediate Corollary of Theorem \ref{thm:1to1} and Proposition \ref{prop:lifeisgood}. For the second point it is enough to consider the fact that, by a standard argument, every weakly Archimedean coherent preference relation can be extended to a complete one, and that $\succ$ is included in $\succ'$ if and only if the SDG corresponding to $\succ$ is included in the SDG corresponding to $\succ'$.
 \end{proof}
 

\bibliographystyle{abbrvnat}
\bibliography{biblio_short}

\end{document}